\documentclass[twoside]{article}
\usepackage{amsmath,amsthm,amssymb}
\usepackage[normalem]{ulem}
\usepackage{newtxtext,newtxmath}
\usepackage[text={12.5cm,19cm},centering,paperwidth=17cm,paperheight=24cm]{geometry}
\usepackage[fontsize=10.4pt]{scrextend}
\usepackage[dvipsnames]{xcolor}
\def\titlerunning#1{\gdef\titrun{#1}}

\makeatletter
\def\author#1{\gdef\autrun{\def\and{\unskip, }#1}\gdef\@author{#1}}
\def\address#1{{\def\and{\\\hspace*{15.6pt}}\renewcommand{\thefootnote}{}\footnote{#1}}\markboth{\autrun}{\titrun}}
\makeatother

\def\email#1{email: \href{mailto:#1}{#1} }
\def\subjclass#1{\par\bigskip\noindent\textbf{Mathematics Subject Classification 2020.} #1}
\def\keywords#1{\par\smallskip\noindent\textbf{Keywords.} #1}

\newenvironment{dedication}{\itshape\center}{\par\medskip}
\newenvironment{acknowledgments}{\bigskip\small\noindent\textit{Acknowledgments.}}{\par}

\newtheorem{thm}{Theorem}[section]
\newtheorem{cor}[thm]{Corollary}
\newtheorem{lem}[thm]{Lemma}
\newtheorem{prob}[thm]{Problem}

\newtheorem{mainthm}[thm]{Main Theorem}

\theoremstyle{definition}

\numberwithin{equation}{section}

\usepackage{graphicx,color}

\newtheorem{prop}[thm]{Proposition}			

\newcommand{\black}{\color{black}}
\newcommand{\blue}{\color{blue}}

\newcommand{\bs}[1]{\blue { #1} \black }

\newcommand{\R}{\mathbb R}
\newcommand{\C}{\mathbb C}

\newcommand{\nc}{\newcommand}
\nc{\e}{\epsilon}
\nc{\be}{\beta}
\nc{\del}{\delta}
\nc{\G}{\Gamma}
\nc{\g}{\gamma}
\nc{\gam}{\gamma}
\nc{\ka}{\kappa}
\nc{\lam}{\lambda}
\nc{\Lam}{\Lambda}
\nc{\Om}{\Omega}
\nc{\om}{\omega}
\nc{\ta}{\tau}
\nc{\w}{\omega}
\nc{\io}{\iota}
\nc{\h}{\theta}
\nc{\z}{\zeta}
\nc{\si}{\sigma}
\nc\vphi{{\varphi}}
\nc\eps{\epsilon}

\newcommand{\lan}{\langle}
\newcommand{\ran}{\rangle}

\newcommand{\one}{\mathbf{1}}

\newcommand{\Null}{\operatorname{Null}}

\newcommand{\Ran}{{\rm{Ran\, }}}

\newcommand{\p}{\partial}

\newcommand{\ra}{\rightarrow}

\newcommand{\VERSIONS}[1]{}

\newcommand{\oP}{P^\perp}

\newcommand{\Pz}{P}

\newcommand{\Pm}{P_{\mu i}}
\newcommand{\oPm}{P_{\mu i}^\perp}

\newcommand{\aaa}{\alpha}

\newcommand{\las}{\lambda_{*}}

\newcommand{\lax}{{\lambda}}

\usepackage[hyperfootnotes=false,colorlinks=true,allcolors=blue]{hyperref}

\begin{document}

\titlerunning{The FS-map and perturbation theory}

\title{\textbf{The Feshbach-Schur map and perturbation theory}}

\author{Genevi\`eve Dusson \and Israel Michael Sigal \and Benjamin Stamm}

\date{}

\maketitle

\address{
G. Dusson: 
Laboratoire de Math\'ematiques de Besan\c{c}on, UMR CNRS 6623, Universit\'e Bourgogne Franche-Comt\'e,  
25030 Besan\c{c}on, France; 
\email{genevieve.dusson@math.cnrs.fr}.  Supported in part by the French ``Investissements d'Avenir'' program, project ISITE-BFC (contract ANR-15-IDEX-0003).
\and 
I.M. Sigal: 
Department of Mathematics,
University of Toronto,
Toronto, ON M5S 2E4, 
Canada; 
\email{im.sigal@utoronto.ca}. Supported in part by NSERC Grant No. NA7901.
\and
B. Stamm: 
Center for Computational Engineering Science, RWTH Aachen University, 52062 Aachen, Germany; 
\email{best@acom.rwth-aachen.de}.
}

\begin{dedication}
\qquad { To Ari with friendship and admiration}
\end{dedication}

\begin{abstract} 
 This paper deals with perturbation theory for discrete spectra of linear operators. To simplify exposition we consider here self-adjoint operators. 
This theory is based on the Feshbach-Schur map and it has advantages with respect to the standard perturbation theory in three aspects: (a) it readily produces rigorous estimates on eigenvalues and eigenfunctions with explicit constants; (b) it is compact and elementary (it uses properties of norms and the fundamental theorem of algebra about solutions of polynomial equations); and (c) it is based on a self-contained formulation of a fixed point problem for the eigenvalues and eigenfunctions, 
allowing for easy iterations. 
 We apply our abstract results to obtain rigorous bounds on the ground states of  Helium-type ions. 
  
\end{abstract}

\section{Set-up and result}
\label{sec:intro}

The eigenvalue perturbation theory is a major tool in mathematical physics and applied mathematics. In the present form, it goes back to Rayleigh and Schr\"odinger and became a robust mathematical field in the works of Kato and Rellich. 
It was extended to quantum resonances by Simon, see \cite{CFKS,HS,Ka,RSII, RSIV, Rell,Sim1, Sim2} for books and a book-size review. 
 
 A different approach to the eigenvalue perturbation problem going back to works of Feshbach and 
 based on the Feshbach-Schur map was introduced in \cite{BFS} 
  and extended in ~\cite{GS, DSS}.
 
 In this paper, we develop further this approach proposing a self-contained theory in a form of a fixed point problem for the eigenvalues and eigenfunctions. 
 It is more compact and direct than the traditional one and, as we show elsewhere, 
  extends to the nonlinear eigenvalue problem.

We show that this approach leads naturally to bounds on the eigenvalues and eigenfunctions with  explicit constants, which we use in an estimation of the ground state energies of the Helium-type ions.  
 
The approach can handle tougher perturbations, non-isolated eigenvalues (see \cite{BFS, GS}) and continuous spectra as well as discrete ones. In this paper, we restrict ourselves to the latter. Namely, we address  the eigenvalue  perturbation problem for operators on a Hilbert space of the form 
\begin{equation}
\label{H}
  H = H_0 +  W,
\end{equation}
where $H_0$ is an operator with some isolated eigenvalues and $W$ is an operator, small relative to $H_0$ in an appropriate norm. The goal is to show that $H$ has eigenvalues near those of $H_0$ and estimate these eigenvalues.

Specifically, with $\|\cdot \|$  standing for the vector and operator norms in the underlying Hilbert space, we assume that 

(A) $H_0$ is a self-adjoint, non-negative operator ($H_0\ge \beta > 0$); 

(B) $W$ is  symmetric and 
form-bounded with respect to $H_0$, in the sense that 
\[
	\|H_0^{-\frac12} W H_0^{-\frac12}\|<\infty.
\] 

(C) $H_0$ has an  isolated eigenvalue $\lam_{0 } > 0$ of a finite multiplicity, $m$.

Here $\| \cdot  \|$ is the operator norm  and   $H_0^{-s}, 0< s < 1,$ is defined either by the spectral theory or by the explicit formula 
\begin{align*} 
   H_0^{-s}:=c \int_0^\infty (H_0 +\om)^{-1} \frac{d\om}{\om^s},
\end{align*} 
where $c :=\left[\int_0^\infty (1+\om)^{-1} \frac{d\om}{\om^s}\right]^{-1}$.

It turns out to be useful in the proofs below to use the following form-norm 
\begin{align*} 
\| W  \|_{\! H_0 } := \|H_0^{-\frac12} W H_0^{-\frac12}\|,
\end{align*}
   
Let $P$ be the orthogonal projection onto the span of  the eigenfunctions
of $H_0$ corresponding to the eigenvalue $\lam_{0 }$ and let
$\oP := \one - P$. 
Let $\gam_{0 }$ be the distance of $\lax_{0 }$ to the rest of the spectrum of $H_0$, and $\las:=\lax_{0 }+\g_{0 }$. In what follows, we often deal with the expression
\begin{align} \label{Phi}
\Phi(W):=\frac{\lax_{0 }\las}{\g_0}\|\oP W P\|_{H_0 }^2. 
\end{align} 
The following theorem proven in Section~\ref{sec:pertth-exten} is the main result of this paper.

\begin{mainthm}  
   \label{thm:FS-pert}
Let Assumptions (A)-(C)  be satisfied and 
assume that, for some $\bs{0<}b<1$ and $\bs{0<}a< 1- b$, 
\begin{align} \label{W-cond} 
\|\oP W\oP\|_{H_0 } &\le \frac{ b\g_0}{\las},\\
\label{W-cond2}\| \Pz W\Pz\|+ k\Phi(W) & < a\g_0,\\
\label{W-cond3}  k\Phi(W) & < \tfrac12 (a\g_0-\|\Pz W\Pz\|),\end{align}
 where $k :=\frac{1}{1-a -b}$. Then the spectrum of the operator $H$ near $\lam_0$ consists of isolated eigenvalues $\lam_i$ 
 of the total   multiplicity $m$  satisfying, together with 
their  normalized eigenfunctions $\vphi_{i}$, the following  estimates 
\begin{align}    \label{lami-est}
|\lax_{ i}-\lax_{0 }| 
& \le   \| \Pz W\Pz\|+ k\Phi(W),\\   \label{EFs-est}   \|\varphi_{ i} - 
    \varphi_{ 0 i} \|  & \le  k \sqrt{\frac{\Phi(W)}{\g_0}}, 
 \end{align}
  where $\varphi_{0 i}$ are appropriate eigenfunctions of $H_0$ corresponding to the eigenvalue $\lax_{0 }$. 
 \end{mainthm}

\noindent {\bf Remark 1} (Comparison with \cite{DSS}). 
A similar result 
 was already proven in \cite{DSS}. Here, 
 the theory is made self-contained and formulated as the fixed point problem and the bounds are tightened. 

\medskip

\noindent  {\bf Remark 2} (Conditions on $W$) The fine tuning the conditions \eqref{W-cond}-\eqref{W-cond3} on $W$ is used in application of this theorem to atomic systems in Section~\ref{sec:Appl}.
 
Note that because of  the elementary estimate 
\[
	\Phi(W) \le \frac{1}{\g_0}\|PW\oP W P \|, 
\]
see  \eqref{PWoP-est1} below, the computation of $\|PW P \|$ and $\Phi(W)$ reduces to computing the largest eigenvalues of the simple $m\times m$-matrices $PW P, - PW P$ and $PW\oP W P$. 

\medskip

\noindent  {\bf Remark 3}  (Higher-order estimates and non-degenerate $\lam_0$).  
In fact, one can estimate $\lax_{ i}$ (as well as the eigenfunctions $\varphi_i$) to an arbitrary order in $\|\oP W P\|_{H_0 }/\g_0$. As a demonstration, we derive (after the proof of Theorem \ref{thm:FS-pert}) the second-order estimate of the eigenvalue in the rank-one ($m=1$) case
\begin{align}
    \label{lami-est'}
|\lax_{ 1}-\lax_{0 }-\lan W\ran |  
 \le 
k\Phi(W), \end{align}
 where 
 $\varphi_{0 }$ is   the eigenfunction of $H_0$ corresponding to $\lax_{0 }$, $\lan W\ran:=\lan \varphi_{0 }, W\varphi_{0 }\ran$. 
 
 For degenerate $\lam_0$, we would like to prove a similar bound on $|\lax_{ i}-\lax_{0 } -\mu_i|$,  where $\mu_{ i}$ is the corresponding eigenvalue of the $m\times m$-matrix $PWP$. 
Here, we have a partial result (proven at the end of the next section) for the lowest eigenvalue $\lax_1$ of $H$:
 \begin{align}\label{lam1-est}\lax_{0}+\mu_{\rm min}-  k\Phi(W) \le \lax_1\le& \lax_{0}+\mu_{\rm min},
 \end{align}
where  $\mu_{\rm min}$ denotes the smallest eigenvalue of the matrix $\Pz W\Pz$. 

\medskip

\noindent  {\bf Remark 4} (Non-self-adjoint $H$). With the sacrifice of the explicit constants in \eqref{lami-est} and \eqref{EFs-est} (mostly coming from \eqref{H1/Hlam-est}), the self-adjointness assumption on  $H$ can be removed. However, for the problem of quantum resonances, one can still obtain explicit estimates.
 
In the rest of this section $H$ is an abstract operator not necessarily self-adjoint or of the form \eqref{H}. Our approach is grounded in the following theorem 
 (see~\cite{GS}, Theorem 11.1). 
   
\begin{thm}\label{thm:isospF}
Let $H$ be an operator on a Hilbert space and $P$ and $P^\perp$, a pair of 
projections such that $P+P^\perp = 1$. Assume $H^\perp:=P^\perp H P^\perp$ is invertible 
on $\Ran P^\perp$ and 
 the expression
\begin{align} \label{Fesh} F_{P} (H ) \ := \ P (H  - H R^\perp H) P ,
\end{align}
where $R^\perp:=P^\perp  (H^\perp)^{-1} P^\perp$, defines a bounded operator. Then $F_{P}$, considered as  a map on the space of operators,  is \textit{isospectral} in the following sense:
\begin{itemize}
\item[(a)] $ \lambda \in \sigma(H ) \quad \mbox{if and only if} \quad 0 \in \sigma(F_{P}
(H - \lambda))$; 
\item[(b)]  
$H\psi = \lambda \psi \quad \mbox{if and only if}  \quad 
F_{P} (H - \lambda) \, \varphi = 0;$
\item[(c)] 
$\dim Null (H - \lambda) = \dim Null F_{P} (H -
\lambda)$.

\item[(d)] $\psi$ and $\varphi$ in (b) are related as $\varphi= P\psi$ and $\psi=Q_{P}(\lax)\varphi$, where 
\begin{align} \label{QP}
    Q_{P}(\lax)   :=  P - P^\perp \, (H^\perp- \lambda)^{-1}P^\perp H P.
\end{align}
\end{itemize}Finally, $F_{P} (H )^*=F_{P^*} (H^* )$ and therefore, if $H$ and $P$ are self-adjoint, then so is $F_{P} (H )$.
\end{thm}
A proof of this theorem is elementary and short; it can be found in \cite{BFS}, Section IV.1, pp 346-348, and \cite{GS}, Appendix 11.6, pp 123-125.

The map $F_{P}$ on the space of operators, is called the {\it Feshbach-Schur map}. The relation $\psi=Q_{P}(\lax)\varphi$ allows us to reconstruct the full eigenfunction from the projected one.
By statement (a), we have
\begin{cor}\label{cor:nuFP} Assume there is an open set $\Lambda\subset \ \C$ such that $H- \lambda$, 
  $\lam\in\Lambda$, is in the domain of the map $F_P$, i.e. $F_P(H- \lambda)$ is well defined. Define the 
   operator-family
 \begin{align*} 
   H(\lam):=F_{P} (H - \lambda)+ \lambda P,
\end{align*} 
and let $\nu_i(\lax)$, for $\lax$ in  $\Lambda$, denote  its eigenvalues counted with multiplicities.
  Then the eigenvalues  of $H$ in  $\Lambda$ are in one-to-one correspondence with the solutions of the equations  \begin{align} \label{nu-fp'}\nu_i(\lax)=\lax.\end{align} 
\end{cor}

Concentrating on the eigenvalue problem, Corollary \ref{cor:nuFP} shows that the original  problem 
\begin{align}
	\label{EVP}
	H\psi = \lam\psi,
\end{align}  is mapped into the equivalent  
 eigenvalue problem,   
\begin{equation}
	\label{EVPeff}
	H(\lax) \varphi = \lam \varphi,
\end{equation}
nonlinear in the spectral parameter $\lax$, but on the smaller space $\Ran P$.

Since the projection $P$ is of a finite rank, the original 
 eigenvalue problem \eqref{EVP}  is mapped into an equivalent lower-dimensional/finite-dimensional one, \eqref{EVPeff}.
 Of course, we have to pay a price for this:  at one step we have to solve a one-dimensional fixed point problem
that can be equivalently seen as a non-linear eigenvalue problem and invert an operator in $\Ran \oP$.

 We call this approach the {\it Feshbach-Schur map method}, or {\it FSM method}, for short.  It is rather compact, as one easily see skimming through   this paper and entirely elementary. 
   
We call $H(\lax)$ the \emph{effective Hamiltonian} (matrix) and write as 

\begin{align}\label{Hlam-deco}
  H(\lax) = P H P +  U(\lam) ,
\end{align}
with the self-adjoint \emph{effective interaction}, or a Schur complement,  $U (\lax)$,  defined  as
\begin{align}
	\label{Ulam-def}
	U(\lam):=
	-P H\oP (H^\perp- \lax)^{-1} \oP HP.
\end{align}  
It is shown in Lemma \ref{lem:FSM-conds} below that \eqref{Ulam-def} defines a bounded operator family.  
  
 We mention here some additional properties discussed in~\cite{DSS}. 
 \begin{prop}\label{prop:U-prop-SA} Let $H$ be self-adjoint, let $\Lambda$ be the same as in Corollary \ref{cor:nuFP}, let $m$ be the rank of $P$ and let the eigenvalues of $H(\lax)$,  $\lax\in \Lambda\cap \R$, be labeled in the order of their increase  counting their multiplicities so that 
\begin{align}
    \label{nui-order}
    \nu_1(\lax) \le \ldots \le \nu_{m}(\lax).
\end{align}
  Then we have that (a)  a solution of the equation $\nu_i(\lax)=\lax$, for  $\lax\in \Lambda\cap \R$,  is the $i$-th eigenvalue, $\lax_i$, of $H$ and vice versa; (b)  ${\nu}_i$ is differentiable in $\lax$ and ${\nu}_i'(\lax)\bs{\le}0$, for $\lax\in \Lambda\cap \R$.  
  \end{prop}
\begin{proof} (a) By \eqref{nui-order}, $\lax_i=\nu_i(\lax_i) <  \nu_{ i+1}(\lax)=\lax_{ i+1}$ (except for the level crossings), which proves the result. For (b), 
by the explicit formula   
\begin{align*}
	U'(\lam):=-P H\oP (H^\perp- \lax)^{-2} \oP HP\le 0,
\end{align*}
we have $U'(\lax)\le 0$, which by the Hellmann-Feynman theorem
 (cf. \eqref{nu'}) implies ${\nu}_i'(\lax) \le 0$.  
\end{proof}

\noindent \textbf{Remark 5} (Perturbation expansion).  In the context of Hamiltonians of  form \eqref{H} satisfying Assumptions (A)-(C), the FSM leads to a perturbation expansion to an arbitrary order. Indeed, in this case, $P$ is the orthogonal projection onto $\Null(H_0-\lam_{0 })$, $\oP := \one - P$ and  $U (\lax)$ can be written as
\begin{align*}
	U(\lam):=
	-P W\oP (H^\perp- \lax)^{-1} \oP WP.
\end{align*}
Now,  using the notation $A^\bot := P^\bot AP^\bot |_{\Ran P^\bot}$ and expanding 
\[
	(H^\perp- \lax)^{-1}= (H_0^\perp+W^\perp- \lax)^{-1}
\] 
in $W^\perp$ and at the same time iterating fixed point equation \eqref{nu-fp'}, we generate a perturbation expansion for eigenvalues $H$ to an arbitrary order (see also Remark 2 above).

The paper is organized as follows. In Section~\ref{sec:pertth-exten}, we present the proof of the main result, Theorem~\ref{thm:FS-pert}. Section~\ref{sec:Appl} 
uses this  result to obtain bounds of the ground-state energy of Helium-type ions.

\section{Perturbation estimates}\label{sec:pertth-exten} 

We want to use  Theorem \ref{thm:isospF}(b, c) to reduce the original eigenvalue problem to a simpler one. In this section, we assume that Conditions (A)-(C) of Section \ref{sec:intro} are satisfied. Recall that $\gam_{0 }$ is the distance of $\lax_{0 }$ to the rest of the spectrum of $H_0$ and the expression $\Phi(W)$ is defined in \eqref{Phi}. 
First, we prove that the operator $F_P(H-\lam)$ is well-defined for $\lax\in \Om$, where, with $a$ the same as in Theorem \ref{thm:FS-pert}, 
\begin{align}\label{Om}\Om:=\{ z\in \C: |z-\lax_{0 }|\le a \g_0\}.\end{align} 
 Recall that $P$ denotes the orthogonal projection onto  $\Null(H_0-\lam_{0 })$ and 
$\oP := \one - P$. Denote  
\[
	H^\bot := P^\bot H P^\bot |_{\Ran P^\bot}
	\qquad \mbox{and} \qquad
	R^\perp(\lam):=P^\bot (H^\bot-\lambda)^{-1} P^\bot
\]  
and recall $k=\frac{1}{1-a -b}$.
We have 
\begin{lem}\label{lem:FSM-conds}
Recall 
$\las:=\lax_{0 }+\g_{0 }$ and assume \eqref{W-cond}. 
Then, for $\lam\in \Om$, the following statements hold
\begin{itemize} \item[(a)]   The operator $H^\bot-\lambda$   is invertible on $\Ran \oP$;
\item[(b)]  
 The inverse $R^\perp(\lam):=\oP(H^\perp- \lax)^{-1}\oP$ defines a bounded, analytic
  operator-family;

\item[(c)]  The expression \begin{equation}\label{Ulam}
U(\lambda):=- 		P H R^\perp(\lam) H P
\end{equation}
 defines a finite-rank, analytic
  operator-family, bounded as
\begin{align}\label{U-bnd}
&\|U(\lax)\|\le k\Phi(W).	
\end{align}
\item[(d)]   
$U(\lax)$ is symmetric for any $\lam\in \Om\cap\R$ and therefore $H(\lax)$ is self-adjoint as well.

\end{itemize} \end{lem}
\begin{proof}  
 With the notation $A^\bot := P^\bot AP^\bot |_{\Ran P^\bot}$, we write
\[
  H^\perp = H_{0}^\perp +  W^\perp.
\]
To prove (a),  we let $H_\lax:=H_0^\perp-\lax$ and use the factorization $H_{\lam}=|H_{\lam}|^{\frac12} V_{\lam}|H_{\lam}|^{\frac12}$, where $V_{\lam}$ is a unitary operator, and 
 use that, for  $\lax\in \Om$, the operator $H_\lax$ is invertible and therefore we have the identity 
 \begin{align}
	\label{H-perp-ident}
	 H^\perp - \lax =|H_{\lam}|^{\frac12} [V_{\lam} + K_{\lam}] |H_{\lam}|^{\frac12},
\end{align}
where $ K_{\lam}:=|H_{\lam}|^{-\frac12} W^\perp |H_{\lam}|^{-\frac12}$.  Next, for  $\lax\in \Om$, we have $\|H_{0}^\perp|H_{\lam}|^{-1} \oP\|=f(\lam )$, where 
\begin{align*}
	f(\lam ):=\sup\bigg\{\left|\frac{s}{s-\lam}\right|: s\ge 0, |s-\lam_0|\ge \g_0\bigg\}. 
\end{align*}
Assuming $ |\lam-\lam_0|\le a\g_0$ and using 
\[
	\left|\frac{s }{s-\lam}\right|
	\le 1+ \left|\frac{\lam }{s-\lam}\right|
	\le 1+\frac{\lam_0+a\g_0 }{(1-a)\g_0}= \frac{\las}{(1-a)\g_0},
\]
we obtain 
\begin{align*}
	f(\lam )\le  \frac{\las}{(1-a)\g_0}.
\end{align*}

 Since $H_{0}^{\frac12}|H_{\lam}|^{-\frac12} \oP=(H_{0}|H_{\lam}|^{-1}|_{\Ran\oP})^{\frac12} \oP$, we have for  $\lax\in \Om$, 
  \begin{align}\label{H1/Hlam-est}
	&\|H_{0}^{\frac12}|H_{\lam}|^{-\frac12} \oP\|\le  \left[\frac{\las}{(1-a)\g_0}\right]^{\frac12},	
\end{align}  
which implies in particular that  $\|K_{\lam}\|\le \frac{ \las}{(1-a)\g_0} \|W^\perp\|_{H_0 }$.  By the assumption~\eqref{W-cond}, i.e., $\|W^\perp\|_{H_0 }$ $\le  \frac{b\g_0}{\las}$, we have    
\begin{align}\label{K-bnd}\|K_{\lam}\|\le \frac b{1-a}.\end{align}
 Since $1-a>b$, by \eqref{H-perp-ident}, the operator $H^\perp- \lax$ is invertible and its inverse is analytic in $\lam\in \Om$, which proves (a) and (b).

We show that statement (c) is also satisfied. Since $P H_0 = H_0 P$ 
  and $P \oP = 0$, we have
 \begin{align*}
  P H \oP =  P W \oP, \quad
  \oP H P = \oP W P.
\end{align*}
These relations and definition \eqref{Ulam} yield
\begin{align} 
  \label{U-expr} &U(\lam) = - P W R^\perp(\lam) W P.
\end{align}
	 Inverting \eqref{H-perp-ident} on $\Ran\oP$ and recalling the notation  $R^\perp(\lam):=\oP(H^\perp- \lax)^{-1}\oP$ gives
 \begin{align}
	\label{resolv-id2}
	R^\perp(\lam)  =|H_{\lam}|^{-\frac12} \oP[V_{\lam} + K_{\lam}]^{-1}\oP |H_{\lam}|^{-\frac12}.
\end{align}

Now, using  identity \eqref{resolv-id2}, estimate \eqref{H1/Hlam-est} and \eqref{K-bnd}, 
we find, for  $\lax\in \Om$, 
 \begin{align}
	\label{ResPerp-bnd2}
	\|H_{0}^{\frac12} R^\perp(\lam)H_{0}^{\frac12}\|\le \frac{k\las}{\g_0}.  
\end{align}
Furthermore, by the eigen-equation $H_0P=\lax_{0 } P$, we have
\begin{align}
	\label{P-est}\|H_0^{\frac12} P\|^2 = \lax_{0}.
\end{align} 
Using expression \eqref{U-expr} and estimates  \eqref{ResPerp-bnd2} and \eqref{P-est}, we arrive at inequality \eqref{U-bnd}. The analyticity follows from \eqref{U-expr} and the analyticity of $R^\perp(\lam)$.

 For (d), since $H_0, W$ and $P$ are  self-adjoint, then so are $U(\lax)$, for any $\lam\in \Om\cap \R$, and, since $U(\lax)$ is bounded, $H(\lax)$ is self-adjoint as well.	\end{proof}

\begin{proof}[ Proof of Theorem~\ref{thm:FS-pert}]  Let $\Om$ be given by equation~\eqref{Om}. 
Recall that, by Lemma \ref{lem:FSM-conds} and Theorem \ref{thm:isospF},  the $m\times m$ matrix-family $H(\lax):=F_P(H  - \lam) + \lam P$, with $F_P$ given in  \eqref{Fesh}, is well defined, for each  $\lax\in \Om$, and  can be written as \eqref{Hlam-deco}. 
 Since $\Pz H \Pz=\lax_0\Pz+ \Pz W\Pz$,  Eq. \eqref{Hlam-deco} can be rewritten as 
\begin{align}\label{Hlam-deco''}H(\lax) = \lax_0\Pz+ \Pz W\Pz +  U(\lam).\end{align} 
Eqs \eqref{U-bnd} and \eqref{Hlam-deco''} imply the inequality
\begin{align}\label{Hlam-est}\|H(\lax) - \lax_0\Pz- \Pz W\Pz\|\,\le \, & 
k\Phi(W).
\end{align}

By a fact from Linear Algebra, for each $\lax\in \Om\cap \R$, the total  multiplicity of the eigenvalues of the $m\times m$ self-adjoint matrix  $H(\lax)$ is $m$.

 Denote  by $\nu_i(\lax), i=1, 2, \dots, m,$ the eigenvalues of $H(\lax)$, {\it repeated according to their multiplicities}.  Eq \eqref{Hlam-est} yields  
 \begin{align}\label{nui-est'}|\nu_i(\lax)-\lax_{0}| \le \, & \| \Pz W\Pz\|+k\Phi(W).
 \end{align}
Indeed, let $P_i(\lax)$ be the orthogonal projection onto  $\Null(H(\lax)-\nu_i(\lax))$. Then \[
	H(\lax)P_i(\lax)=\nu_i(\lax)P_i(\lax),
\]
which, due to \eqref{Hlam-deco''} and $P_i(\lax) P=P_i(\lax)$, can be rewritten as \[(\nu_i(\lax)-\lax_{0})P_i(\lax)=(\Pz W\Pz +  U(\lam))P_i(\lax).\] 
Equating the operator norms of both sides of this equation and using \eqref{Hlam-est} and $\|P_i(\lax)\|=1$ gives \eqref{nui-est'}.

By Corollary \ref{cor:nuFP}, the eigenvalues  of $H$ in the interval $\Om\cap \R$ are in one-to-one correspondence with the solutions of the equations  
\begin{align}\label{nui-eq}
	&\nu_i(\lax)=\lax
\end{align} 
in $\Om\cap \R$. If this equation has a solution, then, due to \eqref{nui-est'}, this solution would satisfy \eqref{lami-est}.
Thus, we address \eqref{nui-eq}. 
Let  
\begin{align}\label{Om'} \Om':=\{ z\in \R: |z-\lax_{0 }|\le r \g_0\},\end{align}
with  
\[
	r:=\frac{1}{\g_0}\big(\| \Pz W\Pz\| + k\Phi(W)\big).
\] 
By our assumption~\eqref{W-cond2}, $r < a<1-b$.

Recall that by the definition, a branch point is a point at which the multiplicity of one of the eigenvalue families (branches) changes. One could think on a branch point as a point where two or more distinct eigenvalue branches intersect. Our next result shows that the eigenvalue branches of $H(\lax)$  could be chosen in a differentiable way and estimates their derivatives.
\begin{prop}\label{prop:H-evs} 
The following statements hold.
\begin{enumerate}
\item[i)] The eigenvalues $\nu_i(\lax)$ of  $H(\lax)$ and the corresponding eigenfunctions can be chosen to be differentiable for $\lax\in\Om'$.
\item[ii)]
 The derivatives $\nu_i'(\lax)$,  $\lax\in \Om'$,
     are bounded as 
\begin{align}	
	\label{nu'-est}	
	|\nu_i'(\lax)| \le  \frac{k}{a-r}\frac{\Phi(W)}{\g_0}.
\end{align} 
\item[iii)]$\nu_i(\lax)$ maps the interval $\Om'$ into itself.
 \end{enumerate}
  Consequently, since by \eqref{W-cond3} the r.h.s. of \eqref{nu'-est} is $<1$, the equations  $\nu_i(\lax)=\lax$ have unique solutions in $\Om'$. 
    \end{prop}

  \begin{proof} 
Proof of (i) for simple $\lax_0$.  For a simple eigenvalue $\lax_0$, $P$ is a rank-one projection on the space spanned by the eigenvector $\varphi_0$ of $H_0$ corresponding to the eigenvalue $\lax_0$ and therefore Eq. \eqref{Hlam-deco''} 
 implies that $H(\lam) =\nu (\lax) P$, with
  \begin{align*}
  &\nu (\lax): =
   \lax_0+\lan \varphi_{0}, (W+  U(\lam))  \varphi_{0}\ran. 
\end{align*}
 This and Lemma \ref{lem:FSM-conds} show that the eigenvalue $\nu (\lax)$  is  analytic.  

Proof of (i) for degenerate $\lax_0$.  We pick an \textit{arbitrary} point $\mu$ in $\Om'$ and 
let $P_{\mu i}$ be the orthogonal projections onto the eigenspaces of $H(\mu)$ corresponding to the eigenvalues ${\nu}_i(\mu)$, 
i.e. 
for a fixed $i$, $P_{\mu i}$ projects on the spans of all eigenvectors with the same eigenvalue ${\nu}_i(\mu)$. We now show that, in a neighbourhood of $\mu$,  
the eigenfunctions $\chi_{i }(\lax)$ can be chosen in a differentiable way. 
We introduce the system of equations for the eigenvalues $\nu_i(\lax)$ 
 and  corresponding eigenfunctions  $\chi_{i }(\lax)$: 
\begin{align}
	\label{nu-lam}&\nu_i(\lax)= \frac{\lan  \chi_{i } (\lax), H(\lax) \chi_{i }(\lax)\ran}{\|\chi_{i }(\lax)\|^2},\\
	\label{chi-lam}	&\chi_{i }(\lax)= \chi_{i }(\mu) - R^\perp({\nu}_i(\lax), \lam)  W_\lam \chi_{i }(\mu),\end{align}
where   
\[
	R^\perp(\nu, \lam):=\oPm (\oPm H(\lax) \oPm -\nu)^{-1}\oPm,  \qquad W_\lam:=U(\lax)- U(\mu).
\]
The expression on the right side of \eqref{chi-lam} is the $Q$-operator for $H(\lam)$ and $P_{\mu i}$ defined according to \eqref{QP} and applied to $\chi_{ i}(\mu)$. Note that  systems \eqref{nu-lam}-\eqref{chi-lam} for different indices $i$ are not coupled. Furthermore, since $\chi_{i }(\mu)$ and $R^\perp({\nu}_i(\lax), \lam)  W_\lam \chi_{i }(\mu)$ are orthogonal, we have $\|\chi_{i }(\lax)\|\ge \|\chi_{i }(\mu)\|$; and $\chi_{i }(\lax)$ and $\chi_{j }(\lax)$ are almost orthogonal for $i \neq j$.  Assuming  this system has a solution $(\nu_i(\lax), \chi_{i }(\lax))$, 
  we see that, by  Theorem \ref{thm:isospF}(d), with $P=\Pm$,   $\chi_{i }(\lax)$ are eigenfunctions of $H(\lax)$ with the eigenvalues $\nu_i(\lax)$ and \eqref{nu-lam} follows from the eigen-equation $H(\lax) \chi_{i }(\lax)=\nu_i(\lax) \chi_{i }(\lax)$.   
 
 For each $i$, we can reduce system \eqref{nu-lam}-\eqref{chi-lam} to the single equation for ${\nu}_i(\lax)$, by treating $\chi_{i }(\lax)$ in \eqref{nu-lam} as given by \eqref{chi-lam}. 
  This leads to the fixed point problem for the functions ${\nu}_i(\lax)$: 
 \begin{align*}
   &\nu =F_{i }(\nu, \lam),
\end{align*} 
where
 \begin{align*}
& F_{i }({\nu}, \lam):=\frac{\lan \chi_{i } (\nu, \lax), H(\lax) \chi_{i }({\nu}, \lax)\ran}{\|\chi_{i }(\nu, \lax)\|^2},\\
   	&\chi_{i }(\nu, \lax) := \chi_{i }(\mu) - R^\perp({\nu}, \lam)  W_\lam \chi_{i }(\mu).\end{align*} 
Notice that  since $\oPm H(\lax) \oPm$ are self-adjoint, 
the resolvent $R^\perp({\nu}, \lam)$ and its derivatives in $\nu$ are uniformly bounded in a  neighbourhood of $(\nu_i(\mu), \mu)$, which does not contain branch 
 points, except, possibly, for $\mu$. 

To be more specific, the resolvent $R^\perp({\nu}, \lam)$ is uniformly bounded 
 in a neighbourhood $\Omega_{\mu i}$  of $(\nu_i(\mu), \mu)$ whether $\mu$ is a branch point or not, as long as $\Omega_{\mu i}$ does not contain any other branch point than $\mu$.

 Hence $F_{i }({\nu}, \lam)$ is differentiable in $\nu$ and $\lam$ and $|\p_\nu F_{i }({\nu}, \lam)|\ra 0$, as $\lam\ra \mu$ (since $\p_\nu \chi_{i }({\nu}, \lax)=R^\perp({\nu}, \lam)^2  W_\lam \chi_{i }(\mu)$ and therefore \[\|\p_\nu \chi_{i }({\nu}, \lax)\|\le \|R^\perp({\nu}, \lam)^2(H_0+ \aaa)^{\frac12}\|\|W_\lam\|_{H_0, \aaa}\|(H_0+ \aaa)^{\frac12} \chi_{i }(\mu)\|\ra 0,\] as $\lam\ra \mu$). Moreover, $\nu_i(\mu) - F_{i }(\nu_i(\mu), \mu)=0$. Let  $f_{i }(\nu, \lam):=\nu-F_{i }({\nu}, \lam)$. Then, by the above,  $f_{i }(\nu_i(\mu), \mu)=0$ and  $|\p_\nu f_{i }(\nu, \lam)-1|\ra 0$, as $\lam\ra \mu$. Thus, the implicit function theorem is applicable to the equation $f_{i }(\nu, \lam)=0$ and shows that there is a unique solution $\nu_{i }(\lax)$ in a  neighbourhood of $(\nu_i(\mu), \mu)$ and that this solution is differentiable in $\lam$.

Next, we have  
 \begin{lem}\label{lem:HFT-ns} 
We have, for $\lax\in \Om'$,   
\begin{align}
	\label{nu'}
	\nu_i'(\lax)= \frac{\lan \chi_{i }(\lax), U'(\lax) \chi_{i }(\lax)\ran}{\|\chi_{i }(\lax)\|^2}.  
 \end{align}     
(Eq. \eqref{nu'} is closely related to the widely used  Hellmann-Feynman theorem.)\end{lem}
  \begin{proof}
   Let $\hat\chi_{i  }(\lax):=\frac{\chi_{i  }(\lax)}{\|\chi_{i }(\lax)\|}$.  Writing equation \eqref{nu-lam} as \[\nu_i(\lax)=\lan  \hat\chi_{i } (\lax), H(\lax) \hat\chi_{i }(\lax)\ran\] and 
differentiating this  with respect to $\lam$, we obtain 
\[
	\nu_i'(\lax)=\lan  \hat\chi_{i } (\lax), H'(\lax) \hat\chi_{i }(\lax)\ran+\eta(\lax),
\]
where $\eta(\lax):=\lan  \hat\chi_{i }'(\lax), H(\lax) \hat\chi_{i }(\lax)\ran+\lan  \hat\chi_{i }(\lax), H(\lax) \hat\chi_{i }'(\lax)\ran$.
Now, moving $H(\lax)$ in the last term to the l.h.s. and using that $H(\lax)\hat\chi_{i }(\lax)= \nu_i(\lax)\hat\chi_{i }(\lax)$, 
 we find   \[\eta(\lax)=\nu_i(\lax)[\lan \hat\chi_{i }'(\lax), \hat\chi_{i }(\lax)\ran+\lan \hat\chi_{i }(\lax), \hat\chi_{i  }'(\lax)\ran]=\nu_i(\lax)\p_\lam \lan \hat\chi_{i }(\lax), \hat\chi_{i  }(\lax)\ran=0,\] which implies \eqref{nu'}.   \end{proof}

To prove Eq. \eqref{nu'-est}, we use formula \eqref{nu'} above and the normalization of $\hat\chi_{i }(\lax)$ to estimate $\nu_i'(\lax)$ as \begin{align}
	\label{nu'-est1}
	|\nu_i'(\lax)| \le 
	\|U'(\lax)\|. \end{align}
To bound the r.h.s. of \eqref{nu'-est1}, we 
use the analyticity of $U(\lax)$ in $\Om$ and estimate \eqref{U-bnd}.
Indeed, by the Cauchy integral formula, we have 
\[
	\|U'(\lam )\|\le \frac{1}{R \g_0}
	\sup_{|\lax'-\lax|=R \g_0} \|U(\lax')\|,
\] 
with $R\le a-r$, so that  $\{\lax'\in \C:|\lax'-\lax|< R \g_0\}\subset \Om$, for $\lam\in \Om'$. 
  This, together with \eqref{U-bnd} and under the conditions of Lemma \ref{lem:FSM-conds}, gives the estimate 
\begin{align}
	\label{U'-bnd}
	&\|U'(\lax)\|\le \frac{k}{a-r} \frac{\Phi(W)}{\g_0}, 
\end{align}
for $\lam\in \Om'$. Combining Eqs. \eqref{nu'-est1} 
  and \eqref{U'-bnd},   we arrive at \eqref{nu'-est}.

Due to  the definition of $r$ after \eqref{Om'}, we can rewrite estimate \eqref{nui-est'} as 
\begin{align*}
   |\nu_i(\lax)-\lax_{0}|&\le r\g_0.
 \end{align*}
By \eqref{Om'}, this shows that $\nu_i(\lax)$ maps the interval $\Om'$ into itself which proves (iii).

 Hence, under condition \eqref{W-cond2}, the fixed point equations $\lax=\nu_i(\lax)$ have unique solutions on the interval $\Om'$, proving Proposition \ref{prop:H-evs}.  \end{proof}

By Corollary \ref{cor:nuFP}, the eigenvalues  of $H$ in  $\Om'$ are in one-to-one correspondence with the solutions of the equations $\nu_i(\lax)=\lax$. 
By  Proposition \ref{prop:H-evs}, 
these equations have unique solutions, say, $\lax_i$. Then, estimate \eqref{nui-est'} implies  inequality \eqref{lami-est}.

To obtain   estimate \eqref{EFs-est}, we 
recall from Theorem~\ref{thm:isospF} that   $Q (\lax_{j})\varphi_{ 0 j}=\varphi_{ j}$,  where  the operator $Q (\lax)$ is given by  
\begin{align*} 
    Q (\lax)  &    := \one - R^\perp(\lam) \oP W P,\end{align*}
 $\lax_{j}$ are the eigenvalues of $H$ 
 and $\varphi_{0j}$ are eigenfunctions of $H_0$ corresponding to $\lam_0$. This gives
 \begin{align}
	\label{EF-differ}\varphi_{ 0 j} - \varphi_{ j}=\varphi_{ 0 j} - Q(\lax_{ j})  \varphi_{ 0j}=R^\perp(\lam_{ j}) \oP W P. 
 \end{align}
Now, as in the derivation of \eqref{ResPerp-bnd2}, using  identity \eqref{resolv-id2}, estimates \eqref{H1/Hlam-est} and 
\[
	\||H_{\lam}|^{-\frac12}\|\le \frac{1}{\sqrt{(1-r)\gam_0}}
\]
and the estimate  $\|K_{\lam}\|\le \frac b{1-a}$, see \eqref{K-bnd}, we find, for  $\lax\in \Om'$,  
\begin{align}
   \label{ResPerp-bnd3}\| R^\perp(\lam)H_{0}^{\frac12}\|\le \frac{k \las^{\frac12}}{\g_0},  
\end{align}
noting that $r<a$.
sing inequalities \eqref{ResPerp-bnd3} and \eqref{P-est}, we estimate the r.h.s. of \eqref{EF-differ} as
\begin{align*} 
& \|R^\perp(\lam_{ j}) \oP W P \|  \le k \sqrt{\frac{\Phi(W)}{\g_0}}.
\end{align*}
This, together with \eqref{EF-differ}, gives \eqref{EFs-est}. This proves Theorem \ref{thm:FS-pert}. \end{proof}

\noindent {\bf Remark 6}.  Differentiating \eqref{chi-lam} with respect to $\lam$, setting $\lam=\mu$, using $W_{\lam=\mu}=0$ and $W_{\lam}'=U'(\lam)$ and changing the notation $\mu$ to $\lax$, we find the following formula for $\chi_{i }'(\lax)$:
 \begin{align}\notag 
 	&\chi_{i }'(\lax)=  - R^\perp({\nu}_i(\lax), \lam)  U'(\lam) \chi_{i }(\lam).\end{align}

Finally we prove the relations \eqref{lami-est'} and \eqref{lam1-est} stated in the introduction.

  \begin{proof}[Proof of \eqref{lami-est'}]  Eq. \eqref{Hlam-deco''} gives 
  \begin{align*}
  &H(\lam) =(\lax_{0}  + \lan W\ran+  \lan U(\lam)\ran)P,
\end{align*} 
  where we use the notation  $\lan A\ran:=\lan \varphi_{0 }, A\varphi_{0 }\ran$.
Since $P$ is a rank one projection, it follows that $H(\lam)$ has only one eigenvalue and this eigenvalue is 
\[\nu_1(\lax)=\lax_{0}+ \lan W\ran+  \lan U(\lam)\ran.\]
This formula and estimate \eqref{U-bnd} show that $\nu_1(\lax)$ obeys the estimate 
\[|\nu_1(\lax)-\lax_{0}-\lan W\ran|\le k\Phi(W).\]

By  Proposition \ref{prop:H-evs}, the eigenvalue $\lax_1$ of $H$ 
 satisfies the equation $\lax_1=\nu_1(\lax_1)$. Then the estimate above implies inequality \eqref{lami-est'}. \end{proof}

 \begin{proof}[ Proof of Eq \eqref{lam1-est}] Let $W$ be symmetric and denote by $\mu_{\rm min}$  the smallest  eigenvalue of the matrix $\Pz W\Pz$. Then $U(\lax)\le 0$ and \eqref{Hlam-deco''} implies $H(\lax) \le \lax_0\Pz+ \Pz W\Pz$. This, together with  $\inf A\le \inf B$ for $A\le B$, gives the upper bound $\nu_1(\lax)\le \lax_{0}+\mu_{\rm min}$ on the smallest eigenvalue-branch  $\nu_1(\lax)$ of $H(\lax)$.
 
For the lower bound, Eqs \eqref{U-bnd} and \eqref{Hlam-deco''} imply the following estimate  
\begin{align*}
   \lax_{0}+\mu_{\rm min}-  k\Phi(W) \le \nu_1(\lax). 
 \end{align*}
The last two estimates and the equation $\nu_i(\lax)=\lax$ (see \eqref{nu-fp'}) yield  \eqref{lam1-est}. \end{proof}  

\section{Application: The ground state energy  of the Helium-type ions} 
\label{sec:Appl}
In this section, we will use the inequalities obtained above to estimate the ground state energy 
 of the Helium-type ions, which is the simplest not completely solvable atomic quantum system. For simplicity, we assume that the nucleus is infinitely heavy, but we allow for a general nuclear charge $ze, z\ge 1$. Then the corresponding  Schr\"odinger operator (describing $2$ electrons of mass $m$ and 
charge $-e$, and the nucleus of infinite mass and
charge $z e$) is given by
\begin{align}\label{hel-ham}
  H_{z, 2} =  \sum_{j=1}^2 \bigg(-\frac{\hbar^2}{2m}\Delta_{x_j} - \frac{\ka z e^2}{|x_j|}\bigg)+  \frac{\ka e^2}{|x_1 - x_2|},
\end{align}
acting on the space $L^2_{sym}({\mathbb R}^{6})$ of $L^2$-functions symmetric (or antisymmetric) w.r.t. the  permutation of $x_1$ and $x_2$\footnote{By the Pauli principle, the product of coordinate and spin wave functions should antisymmetric w.r.t. permutation of the particle coordinates and spins. Hence, in the two particle case, after separation of spin variables,  coordinate wave functions could be either antisymmetric or symmetric.}. Here $\ka ={\frac  {1}{4\pi \varepsilon _{0}}}$ is Coulomb's constant and $\varepsilon _{0}$ is the vacuum permittivity. For $z=2$, $H_{z, 2}$ describes the Helium atom, for $z=1$, the negative ion of the Hydrogen and for $2<z\le 94$ (or $\le 103 $, depending on what one counts as stable elements), Helium-type positive ion. (We can call \eqref{hel-ham} with $z>103$ a \textit{
$z$-ion}.)

It is well known that $H_{z, 2}$ has eigenvalues below its continuum.
Variational techniques give excellent upper bounds on the eigenvalues of $H_{z, 2}$, but good lower bounds are hard to come by. Thus, we formulate

\begin{prob}
Estimate  the ground state energy of $H_{z, 2}$.
\end{prob}

The most difficult case is of  $z=1$, the negative ion of the hydrogen, and the problem simplifies as $z$ increases. 

Here we present fairly precise bounds on the ground state energy of $H_{z, 2}$  implied by our actual estimates. However, the conditions under which these estimates are valid impose rather sever restrictions in $z$.
We introduce the reference energy 
\[
	E^{\rm ref}:=\frac{\ka^2e^4m}{\hbar^2} =\alpha^2 m c^2
\] 
(twice the ground state energy of the hydrogen, or $2$Ry), where $c$  is the speed of light in vacuum and 
\[
 	\alpha =   \ka \frac {e^{2}}{\hbar c}
\]
is the fine structure constant, whose numerical value, approximately  $1/137$. Let $E_{\#}^{(z)}$ stand for either $E_{z, 2}^{\rm sym}/E^{\rm ref}$ or  $E_{z, 2}^{\rm as}/E^{\rm ref}$, the ground state energy  of $H_{z, 2}$ on either symmetric or anti-symmetric functions. We have  
 \begin{prop}  \label{thm:E1-ests} Assume  $z\ge 31.25$ for the symmetric space and $z\ge  170$ for the anti-symmetric one. Then the ground state energy 
 of $H_{z, 2}$ is bounded as 
 \begin{align}    \label{E1-ests}
- c^{\#} z^2 +w_1^{\#} z- \frac{10 w_2^{\#}}{\g_0^{\#}}  \le E_{\#}^{(z)}   \le  - c^{\#}z^2 +w_1^{\#} z, 
 \end{align}
where, for symmetric functions, $c=1$, $\g_0=\frac38$ and  $w_1 \approx 0.6$ and $w_2 \approx 0.27$, defined in equations \eqref{PWP-expr} and \eqref{w2}, and, for anti-symmetric functions, $c^{\rm as}=\frac58$, $\g_0^{\rm as}=\frac{5}{72}$ and   $w_1^{\rm as}\approx 0.20$ and $w_2^{\rm as}\approx 0.01$, defined in \eqref{w1as-w2as}.

 \end{prop}

 The approximate values of $w_1$, $w_2$,  $w_1^{\rm as}$ and $w_2^{\rm as}$ are computed numerically in  Appendix~\ref{NumericalComp}. Here we report the stable digits of our computations.

The inequalities    $z\ge 31.25$ and $z\ge  170$ come from  condition \eqref{W-cond}, while estimates  \eqref{E1-ests}, 
from  \eqref{lami-est'}, with $b=0.8$ and $ a=0.1$ (which give $k=10$).  

Table \ref{tab:per} compares the result for symmetric functions with computations in quantum chemistry. (We did not find results for the antisymmetric space.) 
\\
\begin{table}[h!]
\begin{center}
\begin{tabular}{|l|ccccc|}
\hline
$z$
& 10 & 20 & 30 & 40 & 50 
\\
\hline \hline
$- E_{\rm exact}$ (from \cite{Turb1, Turb2, Turb3})
& 93.9 & 387.7 & 881.4 & 1575.2 & 2468.9 
\\
\hline
main part $- E_{z, 2}^{\rm sym, lead}$ in \eqref{E1-ests} 
& 94 & 388 & 882 & 1576 &  2470 
\\
\hline
relative difference ${\rm err}_z$
& 0.11\% & 0.077\% & 0.068\% & 0.051\% & 0.045\% 
\\
\hline
relative error term $\Delta_{z}$ in \eqref{E1-ests} 
& 8.51\% & 2.06\% & 0.91\% & 0.51\% & 0.32\% 
\\
\hline
\end{tabular}
\end{center}
\caption{Comparison of non-rigorous computations  (see \cite{NakNak, NakNak2, Turb1, Turb2, Turb3}) with the main part $- E_{z, 2}^{\rm sym, lead}:=z^2 -w_1 z$  in equation \eqref{E1-ests}, its relative contribution of the error estimation $\Delta_{z}:=10  w_2/(- E_{z, 2}^{\rm sym, lead}\g_0)$ and the relative difference defined by ${\rm err}_z := |E_{\rm exact} - E_{z, 2}^{\rm sym, lead}|/(- E_{z, 2}^{\rm sym, lead})$.}
\label{tab:per}
\end{table}

We observe that, except for $z=50$, the results of \cite{Turb1, Turb2, Turb3} lie in the interval provided by the estimation~\eqref{E1-ests}.

For computations of the relativistic and QED contributions, see \cite{PachPatkYer, Turb3, YerPach}.

Now, we derive some consequences of  estimates  \eqref{E1-ests}.

Let $z_{*}$ be  the smallest $z$ for which  the error bound in the symmetric case is less than or equal to the smallest explicit (subleading) term $w_1$. Inequality \eqref{E1-ests} shows that the latter bound is satisfied for $z$ such that  $w_1 z \ge \frac{10 w_2}{\gamma_0} $,
which shows that 
 $z_{*} =    \frac{10w_2}{w_1\gamma_0} $ 
  and consequently, 
  \[ z_{*} \approx 12.\]

 According to  equation \eqref{E1-ests}, 
 the  symmetric  ground state energy is lower than the anti-symmetric one, if   $ - z^2+w_1 z< - \frac{5}{8}z^2 + w_1^{\rm as} z-160 \, w_2^{\rm as}$. 
 Using the values  $w_1 \approx 0.6, w_2 \approx 0.3$, $w_1^{\rm as}\approx 0.20$ and $w_2^{\rm as}\approx 0.01$, we find  that, based on  \eqref{E1-ests}, 
 the ground state is symmetric and therefore its spin is $0$ for  
 \[z \ge \frac43\left(\frac25  + \sqrt{\frac{4}{25}+\frac{96}{5}}\right)=\frac83\approx 2.6.\]  
 We conjecture that  the  symmetric  ground state energy is lower than the anti-symmetric one for all $z$'s.

 Finally, note that on symmetric functions, the eigenvalue of the unperturbed operator is simple and  on anti-symmetric ones, has the degeneracy $4$, see below.
 
\begin{proof}[Proof of Proposition \ref{thm:E1-ests}]
First, we rescale the Hamiltonian \eqref{hel-ham} as $x\ra x/\mu$, with 
\[
	\mu:= \frac{z \ka e^2m}{\hbar^2}=\frac{z}{\mbox{the Bohr radius}},  
\]to obtain  
\[
	U_\mu^{-1} H_{z, 2} U_\mu= \frac{z^2 \ka^2e^4m}{\hbar^2} H^{(z)},
\] 
where  $U_\mu: \psi(x)\ra \mu^{3}\psi(\mu x)$ and $H^{(z)}$ is the rescaled  
Hamiltonian given by 
\begin{align*}
  H^{(z)} = \sum_{j=1}^2 \bigg(-\frac{1}{2}\Delta_{x_j} - \frac{1}{|x_j|}\bigg)+  \frac{1/z}{|x_1 - x_2|}.
\end{align*}

Thus, it suffices to estimate  the ground state energy of $H^{(z)}$.  We consider 
\[
	W=\frac{1/z}{|x_1 - x_2|}
\]
as the perturbation and  $\sum_{j=1}^2 (-\frac{1}{2}\Delta_{x_j} - \frac{1}{|x_i|})$ as the unperturbed operator. 

First, we consider the rescaled Hamiltonian $H^{(z)}$ on the \textit{symmetric} functions subspace. On symmetric functions, the ground state energy of $\sum_{j=1}^2 (-\frac{1}{2}\Delta_{x_j} - \frac{1}{|x_j|})$ is $-1$ (see below), so we shift the operator $H^{(z)}$ by $1+\beta$, for some $\beta>0$, 
so that 
\[
	H^{(z)}+1+\beta=H_0+W,
\]
with 
\begin{align*}
H_0 := \sum_{j=1}^2 \bigg(-\frac{1}{2}\Delta_{x_j} - \frac{1}{|x_j|}\bigg)+1+\beta\ \text{ and }\ W=\frac{1/z}{|x_1 - x_2|}.
\end{align*} 
Now, the ground state energy of $H_0$ is $\beta$ and we can use inequality \eqref{lami-est'} and Proposition \ref{prop:U-prop-SA}(c) to estimate  the ground state energy of $H^{(z)}+1+\beta$.

By the HVZ theorem, the spectrum of $H_0$ on symmetric functions consists of the continuum $[e_1+1+\beta, \infty)$ and the eigenvalues $e_m +e_n +1+\beta, m, n \ge 1$, where $e_n, n=1, 2, \dots,$ denote the discrete eigenvalues of the Hydrogen Hamiltonian $-\frac{1}{2}\Delta_{x} - \frac{1}{|x|}$.    
The eigenvalues  $e_n, n=1, 2, \dots,$ are 
 known explicitly:
\begin{align}\label{ezn}  
	e_n = -  \frac{1}{2n^2},
\end{align}
with the multiplicities of $n^2$ and the corresponding  eigenfunctions are given by Eq \eqref{psinlk} below.

Then the ground state energy of $H_0$ is $\lam_0=2e_1+1+\beta= \beta$, as claimed, and the gap, $\g_0= e_1 +e_2 - 2e_1 =e_2  - e_1=  3/8$.

Next, we show that the condition \eqref{W-cond} is satisfied for  $z\ge 31.25$.
We begin with the really rough estimate:
 \begin{align}\label{Wcoul-est'}&	 \frac{1}{|x_1 - x_2|} \le h_{x_1} + h_{x_2}+10, 
  \end{align}
where  $h_x:=-\frac12 \Delta_{x}-\frac{1}{|x|}$. First, we use Hardy's inequality, $-\Delta\ge \frac{1}{4|x|^2}$, and the estimate $\frac{1}{4m|x|^2}-\frac{1}{|x|}\ge -m$ to obtain
\begin{align}\label{Delta-bnd}	& -\frac{1}{m}\Delta-\frac{1}{|x|}\ge  -m.\end{align}
Next,  passing to the relative and centre-of-mass coordinates, $x-y$ and $\frac12(x+y)$, we find
 \begin{align*}
   & - \Delta_x - \Delta_y=-2 \Delta_{x-y}-\frac12 \Delta_{\frac12(x+y)}\ge -2 \Delta_{x-y}.\end{align*}
which, together with \eqref{Delta-bnd}, yields
 \begin{align}	 \frac{1}{|x_1 - x_2|}&\le -\frac12 \Delta_{x_1-x_2}+2\le -\frac14(\Delta_{x_1} + \Delta_{x_2})+2. \notag 
  \end{align}
Now, 
we use $-\frac12 \Delta_{x}=h_x+\frac{1}{|x|}\le h_x-\frac14 \Delta_{x} +4$ to obtain
\begin{align*}
-\frac14 \Delta_{x}\le h_x +4 .
\end{align*}
The last two inequalities yield  \eqref{Wcoul-est'}.
Eq. \eqref{Wcoul-est'}, together with the relation 
\[
	H_0=h_{x_1} + h_{x_2}+1+\beta,
\]
implies the estimate:
 \begin{align}\label{Wcoul-est}&	 \frac{1}{|x_1 - x_2|} \le H_0+9-\beta. 
 \end{align}
 
Now, we check the condition \eqref{W-cond}: 
\[
	\|\oP W\oP\|_{H_0 }\le \frac{b\g_0}{\las}.
\]
Using the last relation, we estimate
 \begin{align*}	z (H_0^\perp)^{-\frac12} W (H_0^\perp)^{-\frac12}&\le (H_0^\perp)^{-1} (H_0^\perp+9-\beta)\\
 &= \one + (9-\beta)(H_{0}^\perp)^{-1}. 
  \end{align*}
 Recalling that $\lam_0=\beta$, we see that $H_0^\perp\ge \beta+\g_0 $, 
 so that the last inequality gives
\[
	0\le  (H_{0}^\perp)^{-\frac12} W (H_{0}^\perp)^{-\frac12} \le 
\frac1z \frac{9+\g_0 }{\beta+\g_0 },
\]
provided $\beta< 9$.
This implies 
\begin{align}\label{oPWoP-est}	\|\oP W\oP\|_{H_0 }\le \frac1z \frac{9+\g_0}{\beta+\g_0 }. \end{align}  Since $\las=\beta+\g_0 $, condition \eqref{W-cond}  is satisfied, if $\frac{9+\g_0 }{z}\le b\g_0 $, which gives $z\ge \frac{9+\g_0}{b\g_0}$. Since $\gam_{0}= 3/8$ and  $b=0.8$,  
 this implies $z\ge 31.25$. 
 We will need the following lemma, whose proof 
 is given after the proof of the proposition. 
 \begin{lem}\label{lem:PWP-PWoP-ests} 
Recall that $P$ is the orthogonal projection onto the eigenspace of $H_0$ corresponding to the lowest eigenvalue 
$\beta$. We have 
    \begin{align}\label{PWP-PWoP-ests}& \lan W\ran= 
     \|PW P \|=\frac{w_1}{z}\quad \text{ and }\quad \Phi(W)\le   \frac{w_2}{z^2\g_0},  \end{align}     
 with, recall, $\lan W\ran:=\lan \varphi_{0 }, W\varphi_{0 }\ran$, where  $\varphi_{0 }$ is   the eigenfunction of $H_0$ corresponding to~$\lax_{0 }$. 
\end{lem}

We check now the second condition of Theorem~\ref{thm:FS-pert}, \eqref{W-cond2}. 
Recall the definition 
 \begin{align*}
r:=\frac{\| \Pz W\Pz\| + k\Phi(W)}{\g_0}.
\end{align*}  
Then condition \eqref{W-cond2} can be written as $r<a$. By Lemma \ref{lem:PWP-PWoP-ests}, 
\[
	r\le \frac1{\g_0}\left(\frac{w_1}{z} + \frac{k w_2}{z^2\g_0}\right).
\] 
Recall the condition $z\ge 31.25$ and  $k=10$, $w_1 \approx 0.6$ and $w_2 \approx 0.3$. Thus
\[
	\frac{k w_2}{z^2 \gamma_0} \le \frac{w_1}{z} \, \frac{ k \, w_2}{31.25 \, \gamma_0 \, w_1}\le  0.5\, \frac{w_1}{z}.
\]
Then we find $r\le  1.5 \, \frac{ w_1}{z\g_0}$ and therefore  condition \eqref{W-cond2} is satisfied if $a> 1.5 \, \frac{w_1}{z\g_0}$, or (for $a=0.1$)
\begin{align*}
	&z> \frac32\frac83 \frac{0.6}{a} = 24. 
\end{align*} 
This is less restrictive than  $z\ge 31.25$.

For the third condition  \eqref{W-cond3}, recalling that $\lax_{0}=\beta$, so that $\las =\gam_{0}+\beta\ (=\lam_1$, the second eigenvalue of $H_0)$  and using 
  both relations in \eqref{PWP-PWoP-ests}, 
   we obtain that condition \eqref{W-cond3} is satisfied if  $  kw_2/(z^2\g_0)< \frac12 (a\g_0-w_1/z)$, which is equivalent to \[z > \frac1{2a\g_0}\left(w_1+\sqrt{w_1^2+8akw_2}\right).\] Using the values  $\gam_{0} =  \frac38$, $w_1 \approx 0.6$ and $w_2 \approx 0.3$, and $b=0.8$ and $ a=0.1$, giving $k=10$, this shows that the latter inequality, and therefore \eqref{W-cond3}, holds if  $z>30$, which is less restrictive than \ $z\ge 31.25$.
Thus, all three conditions are satisfied for  $z\ge 31.25$.

Next, we use \eqref{lami-est'} to estimate the ground state energy, $E_{\rm sym}^{(z)}$ of $H^{(z)}$. 
The first and second relations in \eqref{PWP-PWoP-ests}, together with \eqref{lami-est'} and  the fact that $U(\lam)\le 0$, 
gives the following bounds 
    \begin{align*}  
 - \frac{k w_2}{\gamma_0 z^2} \le E_{\rm sym}^{(z)} +1 -\frac{w_1}{z}   \le  0. 
 \end{align*}
 which after rescaling  
 gives  \eqref{E1-ests}. 
    
Now, we consider the ground state of $H^{(z)}$ on \textit{anti-symmetric} functions. In this case, the ground state energy of $ \sum_{j=1}^2 (-\frac12\Delta_{x_j} - \frac{1}{|x_i|})$ is $e_1 +e_2 =-\frac{5}{8}$ of multiplicity $4$, with the ground states
\begin{equation}
	\label{eq:phi}
	\phi_{2\ell k}(x_1, x_2):=\frac1{\sqrt 2}(\psi _{100}(x_1) \psi _{2\ell k}(x_2) - \psi_{2 \ell k}(x_1) \psi _{100}(x_2)),
\end{equation}
for $(\ell, k)=(0, 0), (1, -1), (1, 0), (1, 1)$,  where $\psi _{n\ell k}(x)$ for $\ell = 0,1,2,\ldots, n-1,  k = -\ell, -\ell+1, \ldots, \ell,
  n = 1, 2, \ldots$,  are  the  eigenfunctions of the Hydrogen-like Hamiltonian $-\frac12\Delta_{x} - \frac{1}{|x|}$ corresponding to eigenvalues $e_n$, so that $\psi _{1 00}(x)=\psi _{1}(x)$, see Appendix~\ref{sec:Hydr-EFs}. Now, we define the unperturbed operator as \[H_{0}:= \sum_{j=1}^2 \Big(-\frac{1}{2}\Delta_{x_j} - \frac{1}{|x_j|}\Big)+\frac{5}{8}+\beta,\]
  to have $\beta$ as the lowest eigenvalue of $H_{0}$.

 For the gap, since the first two eigenvalues are 
 \[
 	e_1+e_2+\frac{5}{8}+\beta=\beta \quad \mbox{and} \quad 
 	e_1+e_3+\frac{5}{8}+\beta (<e_2 +  e_2+\frac{5}{8}+\beta),
\]
we have $\g_0^{\rm as}=e_1+e_3 -  (e_1+e_2)=e_3 -  e_2$ giving $\g_0^{\rm as} =  \frac{5}{72}$.
  
  For condition \eqref{W-cond}, note first that since now $H_0=h_{x_1} + h_{x_2}+\frac58+\beta$,   
  Eq. \eqref{Wcoul-est'} implies 
 \[
 	 \frac{1}{|x_1 - x_2|} \le H_0+9+\frac38-\beta,
 	\]
 	  instead of \eqref{Wcoul-est}, which shows that \eqref{oPWoP-est} should be modified as 
  \begin{align*}
   \|\oP W\oP\|_{H_0 }\le \frac1z \frac{9+\frac38+\g_0^{\rm as}}{\beta+\g_0^{\rm as} }, 
  \end{align*}  
  in the anti-symmetric case.
  Since $\las^{\rm as}=\beta+\g_0^{\rm as} $, condition \eqref{W-cond}  is satisfied, if $\frac{1}{z}(9+\frac38+\g_0^{\rm as})\le b\g_0^{\rm as}$, which gives 
 \[
 	z\ge \frac{9+\frac38+\g_0^{\rm as}}{b\g_0^{\rm as}}.
 \] Since $\gam_{0}^{\rm as}= \frac{5}{72}$ and $b=0.8$, this implies $z\ge 170$ for the anti-symmetric case.  
    
We skip the verification of conditions \eqref{W-cond2} and \eqref{W-cond3}, which are simpler and similar to that of the symmetric case.

  To estimate the ground state energy, $E_{\rm as}^{(z)}$, of $H^{(z)}$ on anti-symmetric functions, we use  bound \eqref {lam1-est}. First, we claim that $\Phi(W)$ defined in \eqref{Phi} satisfies 
  (cf. Remark 2)
     \begin{align}\label{PWoP-est1'}
	&\Phi(W) \le 
	\frac{\mu_{\rm max}}{\g_0^{\rm as}},\end{align}
where $\mu_{\rm max}$ is the largest eigenvalue of the positive, $4\times 4$-matrix $PW\oP W P.$	
Indeed, recall that  
\[
	\Phi(W):= \frac{\lax_{0 }\las\|\oP W P\|_{H_0 }^2}{\g_0^{\rm as}}
\] 
with $\lam_0=\beta$ and observe that $\|PW\oP \|_{H_0}=\|A \|$, with $A:=H_0^{-\frac12} PW\oP H_0^{-\frac12}$. Using now the relations
 \begin{align*}
	AA^*&=H_0^{-\frac12}PW\oP H_0^{-1} WPH_0^{-\frac12}, 
	\end{align*} 
	$H_0^{-\frac12} P=\lax_{0 }^{-\frac12} P$ (since $H_0P=\lax_{0 } P$), and 
$H_0^{-1}\oP \le \las^{-1}\oP$, 
we find 
\[
   AA^*\le (\lax_{0 }\las)^{-1}PW\oP WP,
\]
 which, together with $\|A \|=\|A^* \|=\|AA^*\|^{\frac12}$, gives 
\begin{align}\label{PWoP-est1}
	&\|PW\oP \|_{H_0 }^2
	\le 
	(\lax_{0 }\las)^{-1}\|PW\oP W P \|.\end{align}
Since $PW\oP WP\ge 0$, we have $\|PW\oP W P \|=\mu_{\rm max}$, which gives \eqref{PWoP-est1'}.

Hence,  we have to compute  the smallest  eigenvalue $\mu_{\rm min}$ of the  $4\times 4$ matrix $\Pz W\Pz$ and  the largest eigenvalue, $\mu_{\rm max}$, of the positive, $4\times 4$-matrix $PW\oP W P$,   see \eqref{lam1-est}.   
 
 By scaling, it is easy to see that the matrices $PW P$ and $PW\oP W P$ are of the form $PW P=z^{-1} A$ and $PW\oP W P=z^{-2} B$, where $A$ and $B$ are $z$-independent $4\times 4$ matrices. Hence 
  \begin{align}\label{w1as-w2as}
  	\mu_{\rm min}=\frac{ w_1^{\rm as}}{z}\quad \text{ and }\quad \mu_{\rm max}=\frac{w_2^{\rm as}}{z^2}, 
  	\end{align}
   for some positive $z$-independent constants $w_1^{\rm as}$ and $w_2^{\rm as}$. These constants are computed in Appendix \ref{NumericalComp}.
 Now, we see that  \eqref{lam1-est} and \eqref{PWoP-est1'} give 
 \begin{align*}
 \frac{w_1^{\rm as}}{z}-  \frac{ k w_2^{\rm as}}{\g_0^{\rm as} z^2} \le E_{\rm as}^{(z)}+\frac{5}{8}\le& \frac{w_1^{\rm as}}{z},
 \end{align*}
which after 
 rescaling and setting $k=10$ gives  \eqref{E1-ests} in the antisymmetric case. 
\end{proof}

  \begin{proof}[Proof of Lemma \ref{lem:PWP-PWoP-ests}]  The ground state energy, $e_1$, of $-\frac{1}{2}\Delta_{x} - \frac{1}{|x|}$ is non-degenerate, with the normalized ground state 
  (known as the  $1\mathrm {s}$ wavefunction)  given  by (see e.g. \cite{Mess}, or Wiki, with $a_0$ replaced by $1$ due to the rescaling above)
\begin{align*}
\psi _{1}(x)={\frac {1}{{\sqrt {\pi }}}}e^{-|x|}. 
\end{align*}
  Hence, the ground state energy, $0$, of $H_0$ is also non-degenerate, with the normalized ground state $\psi _{1}(x _{1})\psi _{1}(x _{2})$ (in the symmetric subspace).

First, we compute $\lan W\ran $ and  $\|PW P \|$. Let $(\psi \otimes \phi)(x, y):=\psi (x) \, \phi(y)$ and, for an operator $A$ on $L^2({\mathbb R}^{6})$, we denote $\lan A\ran := \lan \psi_1 \otimes \psi_1, A\psi_1 \otimes \psi_1\ran$. In the symmetric case, $PW P=\lan W\ran P$, where 
\begin{align}\label{PWP-expr}	&\lan W\ran := \lan \psi_1 \otimes \psi_1, W\psi_1 \otimes \psi_1\ran = \frac1z\int_{\R^6}	\frac{ | \psi_{1 }(x)\psi_{1 }(y)|^2}{|x-y|} \; dx\, dy=:\frac1z w_1.\end{align} 
 Hence $\|PW P \|=\frac{w_1}{z}$. This gives the first relation in \eqref{PWP-PWoP-ests}.

		Now, using that $\oP=\one-P$, we compute
	\begin{align}\label{PWoPWP-expr}	& PW\oP W P=PW^2 P - PW PW P  = \frac{w_2}{z^2} P,
   \end{align}
and
   \begin{align}
& w_2 := z^2\big(\lan  W^2 \ran- \lan W \ran^2\big) \nonumber\\
\label{w2}	& \quad =\int_{\R^6}	\frac{ | \psi_{1 }(x)\psi_{1 }(y)|^2}{|x-y|^2} \; dx\, dy-\bigg(\int_{\R^6}	\frac{ | \psi_1(x)\psi_1(y)|^2}{|x-y|} \; dx\, dy\bigg)^2. \end{align}
By the Schwarz inequality and the normalization of $\psi_1$, we have $w_2\ge 0$, which, together with \eqref{PWoPWP-expr}, gives
\begin{align}\label{PW2P-est1}	&
\|PW\oP W P \|=	 \frac{w_2}{z^2} .
\end{align}
Eqs. \eqref{PWoP-est1} and \eqref{PW2P-est1} imply the second relation in \eqref{PWP-PWoP-ests}.	\end{proof}

\begin{acknowledgments}  
   The second author is grateful to Volker Bach and J\"urg Fr\"ohlich for enjoyable collaboration. 
   The authors are indebted to the anonymous referee for many useful suggestions and remarks.  
\end{acknowledgments}

\appendix

\section{The  eigenfunctions of the Hydrogen-like Hamiltonian} \label{sec:Hydr-EFs}

The  eigenfunctions   $\psi _{n\ell k}(x)$   are given by (see Wiki, ``Hydrogen atom'' with $a_0$ replaced by $1$) 
 \begin{align}\label{psinlk} \psi _{n\ell k}(x)=\sqrt {{\left({\frac {2 z}{n}}\right)}^{3}{\frac {(n-\ell -1)!}{2n(n+\ell )!}}}e^{-\rho /2}\rho ^{\ell }L_{n-\ell -1}^{2\ell +1}(\rho )Y_{\ell }^{k}(\theta ,\phi),
\end{align}
for $\ell = 0,1,2,\ldots, n-1,  k = -\ell, -\ell+1, \ldots, \ell,
  n = 1, 2, \ldots$. Here $(r,\theta ,\phi)$ are the polar coordinates of $x$,  ${\displaystyle \rho ={2z r \over {n}}}$,   ${\displaystyle L_{n-\ell -1}^{2\ell +1}(\rho )}$ is a generalized Laguerre polynomial of degree  ${\displaystyle n-\ell -1},$ and
${\displaystyle Y_{\ell }^{k}(\theta ,\phi )}$ is a spherical harmonic function of the degree 
$\ell$  and order  $k$.
\footnote{Quoting from Wikipedia (\url{https://en.wikipedia.org/wiki/Hydrogen_atom}, 30.10.2020): ``... the generalized Laguerre polynomials are defined differently by different authors. The usage here is consistent with the definitions used by \cite{Mess},  p. 1136, and Wolfram Mathematica. In other places, the Laguerre polynomial includes a factor of  ${\displaystyle (n+\ell )!}$. 
...  or the generalized Laguerre polynomial appearing in the hydrogen wave function is  ${\displaystyle L_{n+\ell }^{2\ell +1}(\rho )}$ instead.''
} 

\section{The numerical approximation of the constants $w_1$, $w_2$, $w_1^{as}$ and $w_2^{as}$}
\label{NumericalComp}

Let us first focus on $w_1$ and $w_2$.
From the definition 
\[
	w_1 := \int_{\R^6}	\frac{ | \psi_{1 }(x)\psi_{1 }(y)|^2}{|x-y|} \; dy\, dx, 
\]
\[
	w_2 := \int_{\R^6}	\frac{ | \psi_{1 }(x)\psi_{1 }(y)|^2}{|x-y|^2} \; dy\, dx-w_1^2,
\]
it becomes clear that we need to compute terms of the form
\begin{align*}
   C_\alpha  & = \int_{\R^6}	\frac{ | \psi_{1 }(x)\psi_{1 }(y)|^2}{|x-y|^\alpha} \; dy\, dx \\
	& =
	\int_{\R^3} | \psi_{1 }(x)|^2	\int_{\R^3} \frac{|\psi_{1 }(x-z)|^2}{|z|^\alpha} \; dz\, dx,
\end{align*}
with $\alpha=1,2$.
We thus introduce a numerical quadrature in order to approximate the values of these integrals. 
We do not aim within this work to obtain the most efficient implementation.
We define a numerical quadrature grid in terms of spherical coordinates with integration points $x_{i,n}$ and weights $\omega_{i,n}$ given by
\[
	x_{i,n} = r_i \, s_n, 
	\qquad 
	\omega_{i,n} = h \, r_{i}^3 \, \omega_n^{\rm leb}
\]
with 
\[
r_i =\exp(-\ln(R_{\rm max})+(i-1)h), \qquad i=1,\ldots,N_r, \qquad h= \frac{2\ln(R_{\rm max})}{N_r-1}
\] and where $\{s_n,\omega^{\rm leb}_n\}$ denote Lebedev integration points on the unit sphere with $N_{\rm leb}$ points.
Thus, $N_r$ denotes the number of points along the $r$-coordinate, $R_{\rm max}$ the radius of the furthest points away from the origin.

We thus approximate $C_\alpha$ by
\[
	C_\alpha 
	\approx
	\sum_{i=1}^{N_r} 
	\sum_{n=1}^{N_{\rm leb}}
	\omega_{i,n}
	| \psi_{1 }(x_{i,n})|^2
	\Phi_\alpha(x_{i,n})
\]
with
\[
	\Phi_\alpha(x_{i,n})
	=
	\sum_{j=1}^{N_r} 
	\sum_{m=1}^{N_{\rm leb}}
	\omega_{j,m}
	\frac{|\psi_{1 }(x_{i,n}-x_{j,m})|^2}{|x_{j,m}|^\alpha}.
\]
Then, we approximate $w_1\approx C_1$ and $W_2\approx C_2-C_1^2$.

We now shed our attention to the constants $w_1^{\rm as}$ and $w_2^{\rm as}$ and we first introduce the bi-electronic integrals, for $(\ell, k),(\ell', k')=(0, 0), (1, -1), (1, 0), (1, 1)$, by
\begin{align*}
	{\sf A}_{\ell,k}^{\ell',k'}
	&= 
	\int_{\R^3} \psi_{2\ell k }(x)\psi_{2\ell' k' }(x)	\int_{\R^3}\frac{ | \psi_{1 }(y)|^2}{|x-y|} \; dy\, dx,
	\\
	{\sf B}_{\ell,k}^{\ell',k'}
	&= 
	\int_{\R^3} \psi_{1 }(x)\psi_{2\ell' k' }(x)	\int_{\R^3}\frac{  \psi_{1 }(y)\psi_{2\ell k }(y)}{|x-y|} \; dy\, dx.
\end{align*}
Then, using the definition~\eqref{eq:phi} and symmetry properties, we derive the following expressions for the matrix elements
\[
	{\sf M}_{\ell,k}^{\ell',k'}
	= z\, \lan \psi_{2\ell k } |W| \psi_{2\ell' k' } \ran 
	= {\sf A}_{\ell,k}^{\ell',k'} - {\sf B}_{\ell,k}^{\ell',k'}.
\]
We then employ the same quadrature rule as above to approximate the matrix elements ${\sf A}_{\ell,k}^{\ell',k'} $, ${\sf B}_{\ell,k}^{\ell',k'}$ and compute the smallest eigenvalue of $\sf M$ as an approximation of $w_1^{\rm as}$. 

Finally, the approximation of $w_2^{\rm as}$ involves the matrix 
\[
	PW\oP W P = PW^2P - P W P WP ,
\]
whose elements are given by ${\sf Q}  = {\sf N} -{\sf M}^2$ with
\[
	{\sf N}_{\ell,k}^{\ell',k'}
	= z^2\, \lan \psi_{2\ell k } |W^2| \psi_{2\ell' k' } \ran 
	= {\sf C}_{\ell,k}^{\ell',k'} - {\sf D}_{\ell,k}^{\ell',k'}
\]
and 
\begin{align*}
	{\sf C}_{\ell,k}^{\ell',k'}
	&= 
	\int_{\R^3} \psi_{2\ell k }(x)\psi_{2\ell' k' }(x)	\int_{\R^3}\frac{ | \psi_{1 }(y)|^2}{|x-y|^2} \; dy\, dx,
	\\
	{\sf D}_{\ell,k}^{\ell',k'}
	&= 
	\int_{\R^3} \psi_{1 }(x)\psi_{2\ell' k' }(x)	\int_{\R^3}\frac{  \psi_{1 }(y)\psi_{2\ell k }(y)}{|x-y|^2} \; dy\, dx.
\end{align*}
We then again, approximate the different integrals with the above introduced quadrature rule and find the maximal eigenvalue of $\sf Q$ in order to approximate~$w_2^{\rm as}$.

In Figure~\ref{fig:fig}, we plot the the approximate value of $w_1$, $w_2$, $w_1^{\rm as}$ and $w_2^{\rm as}$ for different values of the parameters $R_{\rm max}$, $N_r$ and $N_{\rm leb}$ given in Table~\ref{tab:quad'}.

\begin{figure}[t!]
	\centering
    \includegraphics[trim = 0mm 0mm 0mm 0mm, clip,width=0.75\textwidth]{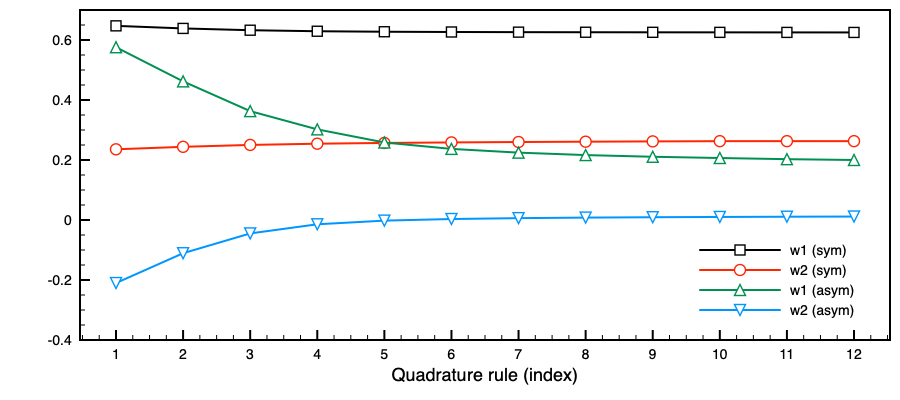}
    \caption{
    Value of the approximate coefficients $w_1$, $w_2$, $w_1^{\rm as}$ and $w_2^{\rm as}$ depending on the quadrature rule reported in Table~\ref{tab:quad'}.
 	}
	\label{fig:fig}
\end{figure}

We obtain approximate values 
\[
	w_1\approx 0.625,\quad w_2 \approx 0.2628, \quad w_1^{\rm as} \approx 0.3, \quad w_2^{\rm as} \approx 0.01
\]
 where the last reported digit is stable over the last three quadrature rules.
\\
\begin{table}[h!]
\begin{center}
\scriptsize
\begin{tabular}{|c|cccccccccccc|}
\hline
index
& 1 & 2 & 3 & 4 & 5 & 6 & 7 & 8 & 9 & 10 & 11 & 12
\\
\hline \hline
$R_{\rm max}$
& 16 & 18 & 20 & 22 & 24 & 26 & 28 & 30 & 32 & 34 & 34 & 34
\\
\hline
$N_r$
& 10 & 12 & 14 & 16 & 18 & 20 & 22 & 24 & 26 & 28 & 30 & 32
\\
\hline
$N_{\rm leb}$ 
& 194 & 266 & 350 & 590 & 974 & 1454 & 2030 & 2702 & 3470 & 4334 & 5294 & 5810
\\
\hline
\end{tabular}
\end{center}
\caption{Parameters for the different quadrature rules.}
\label{tab:quad'}
\end{table}

\newpage


\begin{thebibliography} {CFKS}

\bibitem{BFS}V. Bach, J. Fr\"ohlich  and I.M. Sigal,  Quantum electrodynamics of confined non-relativistic particles. Adv. Math., 137(2):299--395, 1998.


\bibitem{CFKS}
H. Cycon, R. Froese, W. Kirsch, B. Simon,
{\it Schr\"odinger Operators (with Applications
to Quantum Mechanics and Global Geometry}.
Springer, 1987. 

\bibitem{DSS}  G. Dusson, I.M. Sigal and B. Stamm, Analysis of the Feshbach-Schur method for the planewave discretizations of Schr\"odinger operators. 	arXiv:2008.10871.


\bibitem{GS} S. Gustafson and I.M. Sigal, {\it Mathematical Concepts of Quantum Mechanics}. Springer Science \& Business Media, Sept. 2011.

\bibitem{HS}
P. Hislop and I.M. Sigal,
{\it Introduction to Spectral Theory.
With Applications to Schr\"odinger Operators}.
Springer-Verlag, 1996.


\bibitem{Ka} T. Kato, {\it Perturbation theory for linear operators}. Springer-Verlag, 
 1976.

\bibitem{Mess} A. Messiah,  {\it Quantum Mechanics}. New York: Dover. (1999). 


\bibitem{NakNak}H. Nakashima, H. Nakatsuji,
Solving the Schr\"odinger equation for helium atom and its isoelectronic ions with the free iterative complement interaction (ICI) method,
J. Chem. Phys. 127, 224104 (2007). 

\bibitem{NakNak2} H. Nakashima, H. Nakatsuji, Solving the electron-nuclear Schr\"odinger equation of helium atom and its isoelectronic ions with the free iterative-complement-interaction method.
J. Chem. Phys. 128, 154107 (2008).



\bibitem{PachPatkYer} K. Pachucki, V. Patkos and V.A. Yerokhin,
Testing fundamental interactions on the helium atom. Phys. Rev. A95 (2017) 062510.


\bibitem{RSII}
M. Reed and B. Simon,
{\it Methods of Modern Mathematical Physics, Vol. II.
Fourier Analysis and Self-Adjointness.}
Academic Press, 1975.



\bibitem{RSIV}
M. Reed and B. Simon,
{\it Methods of Modern Mathematical Physics IV:
Analysis of Operators}.
Academic Press, 1979.

\bibitem{Rell}
F. Rellich, {\it Perturbation Theory of Eigenvalue Problems}.
Gordon and Breach:  New York, 1969.

\bibitem{Sim1} B. Simon, A Comprehensive Course in Analysis, Part 4: Operator Theory (American Mathematical Society, Providence, RI, 2015).

\bibitem{Sim2} B. Simon, Tosio Kato's work on non-relativistic quantum mechanics: Part 1 and 2, Bulletin of Mathematical Sciences, Vol 8 (2018), 121--232 and Vol 9 (2019) 1950005 (105 pages).



\bibitem{Turb1} A. V. Turbiner  and J. C. L. Vieyra, Helium-like and Lithium-like ions: Ground state energy.
arXiv:1707.07547v3, 2017.


\bibitem{Turb2} A. V. Turbiner, J. C. L. Vieyra, J.C. del Valle, and DJ Nader, Ultra-compact accurate wave functions for He-like and Li-like iso?electronic sequences and variational calculus: I. Ground state.
International Journal of Quantum Chemistry, e26586, 2020; 
  A. V. Turbiner, J. C. L. Vieyra and J.C. del Valle, Ultra-compact accurate wave functions for He-like and Li-like iso-electronic sequences and variational calculus. I. Ground state.
arXiv:2007.11745v, 2020.


\bibitem{Turb3} A.V. Turbiner, J.C. L. Vieyra, H. Olivares-Pil\'on, Few-electron atomic ions in non-relativistic QED: The ground state.
Annals of Physics 409 (2019) 167908



\bibitem{YerPach} V.A. Yerokhin and K. Pachucki,
Theoretical energies of low-lying states of light helium-like ions. Phys Rev A81 (2010) 022507.


\end{thebibliography}
\end{document}